%% file: main.tex
\documentclass[sigconf, nonacm]{acmart}
\usepackage{csquotes}
\usepackage{booktabs}
\usepackage{enumitem}
\usepackage{caption}
\usepackage{subcaption}
\usepackage{wrapfig}
\usepackage{siunitx}
\usepackage{multirow}
\usepackage{flushend}

\usepackage{color}
\usepackage{colortbl}

\DeclareMathOperator{\var}{\mathbf{var}}
\DeclareMathOperator{\sort}{sort}
\DeclareMathOperator{\Def}{Def}
\DeclareMathOperator{\Dom}{Dom}
\DeclareMathOperator{\INF}{INF}
\DeclareMathOperator{\Pos}{Pos}
\DeclareMathOperator{\E}{E}
\newcommand{\eps}{\varepsilon}
\newcommand{\from}{\leftarrow}
\newcommand{\mcA}{\mathcal A}
\newcommand{\mcB}{\mathcal B}
\newcommand{\mcP}{\mathcal P}
\newcommand{\mcX}{\mathcal X}
\newcommand{\N}{\mathbb N}

\usepackage[english]{babel}
\usepackage{amsthm}

\theoremstyle{remark}
\newtheorem{theorem}{Theorem}[section]
\newtheorem{example}[theorem]{Example}
\newtheorem{definition}[theorem]{Definition}
\newtheorem{proposition}[theorem]{Proposition}
\newtheorem{lemma}[theorem]{Lemma}

\begin{document}
\copyrightyear{2024}

\title{A Plaque Test for Redundancies in Relational Data [Extendend~Version]}

\author{Christoph Köhnen}
\affiliation{%
  \institution{Chair of Scalable Database Systems\\University of Passau}}
\email{christoph.koehnen@uni-passau.de}

\author{Stefan Klessinger}
\affiliation{%
  \institution{Chair of Scalable Database Systems\\University of Passau}}
\email{stefan.klessinger@uni-passau.de}

\author{Jens Zumbrägel}
\affiliation{%
  \institution{Professorship of Cryptography\\University of Passau}}
\email{jens.zumbraegel@uni-passau.de}

\author{Stefanie Scherzinger}
\affiliation{%
  \institution{Chair of Scalable Database Systems\\University of Passau}}
\email{stefanie.scherzinger@uni-passau.de}

\begin{abstract}
Inspired by the visualization of dental plaque at the dentist's office, this article proposes a novel visualization of redundancies in relational data. 
Our approach is based on a well-principled information-theoretic framework that has so far seen limited practical application in systems and tools.  In this framework, we quantify the information content (or entropy) of each cell in a relation instance given a set of functional dependencies.  The entropy value signifies the likelihood of recovering the cell value based on the dependencies and the remaining tuples.  By highlighting cells with lower entropy, we effectively visualize redundancies in the data.
We present an initial prototype implementation and demonstrate that a straightforward approach is insufficient to handle practical problem sizes.  To address this limitation, we propose several optimizations which we prove to be correct.  In addition, we present a Monte Carlo approximation with a known error, enabling a computationally tractable analysis.  By applying our visualization technique to real-world datasets, we showcase its potential.  Our vision is to empower data analysts by directing their focus in data profiling toward pertinent redundancies, analogous to the diagnostic role of a plaque test at the dentist's office.

\end{abstract}

\maketitle

\section{Introduction}

Database normalization is a well-studied field, with the theory of functional dependencies as a cornerstone, cf.~\cite{DBLP:books/aw/AbiteboulHV95}. 
There are several proposals to visualize functional dependencies, such as sunburst diagrams or graph-based visualizations~\cite{10.1145/3132847.3133180}.
Yet these approaches only visualize the dependencies, irrespective of the data instance.

We propose a novel visualization that reveals the redundancies captured by functional dependencies. We refer to our approach as \enquote{plaque test}: Like a plaque test at the dentist, which colors dental plaque to reveal unwanted residue on patients' teeth, our plaque test reveals redundancies in relational data.
The plaque test may be applied in data profiling, when dependencies have been discovered,  before carrying out data cleaning tasks, or when preparing towards schema normalization --- in all these scenarios,  domain experts will want to explore the redundancies resident in their data. 

Formally, our approach is based on an existing information-theoretic framework developed by Arenas and Libkin~\cite{DBLP:conf/pods/ArenasL03}.  To our knowledge, we are the first to implement and apply this framework to a practical problem.

Next, we illustrate our concept of plaque tests.

\begin{example}
Figure~\ref{fig:relation} shows an example from the German Wikipe\-dia site on database normalization\footnote{
Retrieved from \url{https://de.wikipedia.org/wiki/Normalisierung_(Datenbank)}, last accessed September 2024.}. 
The relation manages a CD collection.  Each CD has an identifier (ID), an album title (Album), and was released in a given year (RYear).  The band that released the CD (Band) was also founded in a given year (BYear).  Each CD has several tracks (Track) and each track has a title (Title).

We consider the following functional dependencies (not stated as a canonical cover):
\begin{align*}
\mbox{ID} &\rightarrow \mbox{Album, Band, BYear, RYear}\\
\mbox{ID, Track} &\rightarrow \mbox{Title} \qquad \mbox{and} \qquad 
\mbox{Band} \rightarrow \mbox{BYear}
\end{align*}
\input{example_figure.tex}
Figure~\ref{fig:relation2} shows the plaque test applied to the instance in Figure~\ref{fig:relation}, based on these dependencies: Each cell is assigned an entropy value. Our plaque test then visualizes these entropy values by coloring the cells: the more redundant the value, the smaller the entropy value and the deeper the blue.

Specifically, let us consider the dependency \enquote{$\mbox{ID} \rightarrow \mbox{Album}$}.
In the given data instance, the album title \enquote{Not That Kind} with ID~1 is recorded redundantly.
If this information was lost in the first tuple, it could still be recovered from the second or third tuple (and vice versa).
This is captured by an entropy of~0.8 and blue color (denoting \enquote{plaque}). 
In contrast, if the album title \enquote{Freak of Nature} with ID~3 was lost, it could not be recovered. Consequently, the cell has entropy~1 and remains white (\enquote{no plaque}).

Let us next consider the dependency \enquote{Band $\rightarrow$ BYear}, stating that the band determines the year of foundation.
Since there are four tuples for the band Anastacia, the year of the foundation is more redundant than the year of the release of her album \enquote{Not That Kind} (recorded in three tuples). 
Accordingly, the entropy values differ and our plaque test colors the foundation years for Anastacia's band in a deeper blue. 

Finally, let us consider the dependency \enquote{ID $\rightarrow$ BYear}. Since we have three entries for the album with ID 1 (\enquote{Not That Kind}), and only one entry each for the albums with ID 2 (\enquote{Wish You Were Here}) and ID 3 (\enquote{Freak of Nature}), the band year is even more redundant for Anstacia's album \enquote{Not That Kind} (entropy value~0.6) than for her album \enquote{Freak of Nature} (entropy value~ 0.7). 
\end{example}

So far, we only considered genuine dependencies. Next, we assume that the dependencies have been automatically discovered during data profiling.

\begin{example}
The Metanome tool~\cite{DBLP:journals/pvldb/PapenbrockBFZN15} discovers~23 dependencies in this relation instance., including the cyclic dependencies \enquote{Band $\rightarrow$ BYear} and \enquote{\mbox{BYear} $\rightarrow$ \mbox{Band}}.

Figure~\ref{fig:relation3} shows the result of our plaque test, given this new set of functional dependencies.  Throughout, the coloring of the cells is darker.  Moreover, more cells are colored. 
This reveals that the instance now contains additional redundancies.

Notably, the entropy value for the band name \enquote{Anastacia} is as low as~0.4, since it may be recovered from several dependencies, such as \enquote{ID $\rightarrow$ Band} and \enquote{BYear $\rightarrow$ Band}.
Thus, plaque is additive when an attribute occurs on the right-hand side of several functional dependencies.
\end{example}

As these examples illustrate, the plaque test is highly sensitive to the set of functional dependencies and the specific values that appear in the data instance. 

\subsubsection*{Contributions}

Based on an existing information-theoretic framework,
we propose a \enquote{plaque test} as a novel visualization of redundancies in relational data. This is an extended version of~\cite{shortversion}.

\begin{itemize}

\item We show that a straightforward implementation for computing the entropy values underlying our plaque test does not scale beyond toy examples. 

\item Thus we propose several effective optimizations:
\begin{enumerate}
\item We discuss scenarios where entropy values can be immediately assigned to~1, thereby skipping computation, and where one can focus on a subset of the data instance when computing entropy values.  We prove the correctness of these optimizations.
\item We present a Monte Carlo approximation to compute entropy values.  We provide the formula to determine the number of iterations required to achieve a given accuracy with a certain confidence.
\end{enumerate}

\item We present visual plaque tests for several real-world datasets and discuss how the discovered \enquote{plaque} can indeed be helpful in data exploration. 

\item We conduct run-time experiments with our implementation and explore the effect of our optimizations. We show that our optimizations are indeed effective, but also that further optimizations are required if our approach is to scale for the sizes of real-world database instances.
\end{itemize}

\subsubsection*{Reproducibility}

Our research artifacts, including our prototype implementation, are available at~\url{https://doi.org/10.5281/zenodo.8220684}.

\subsubsection*{Structure}

We provide preliminaries on functional dependencies, entropy, and information content in Section~\ref{sec:prelim}.  In Section~\ref{sec:opt}, we introduce two lines of optimization, one is an exact method and the other an approximation.  We present our experiments in Section~\ref{sec:experiments}.  We discuss related work in Section~\ref{sec:related} and conclude in Section~\ref{sec:outlook}.

\section{Preliminaries}
\label{sec:prelim}

When introducing functional dependencies in Section~\ref{ssec:fds}, we deviate from the common definition (cf.~\cite{DBLP:books/aw/AbiteboulHV95}) in that we take into account the order of tuples. This allows us to identify individual cells in a relation instance.
The notion of entropy-related information content originates from the work of Arenas and Libkin~\cite{DBLP:conf/pods/ArenasL03} (Section~\ref{ssec:entropy}).
As a first contribution, we present simplifications to compute the corresponding entropy values (Section~\ref{ssec:simply}).

\subsection{Functional Dependencies}\label{ssec:fds}

Denote by~$\N$ the set of positive integers, $\N := \{ 1, 2, \dots \}$.
\begin{definition}
    A \emph{relation}~$R$ of arity~$m$ is specified by a finite set $\sort(R) := \{ A_1, \dots, A_m \}$ of attributes~$A_i$.  The domain of an attribute~$A$ is denoted by $\Dom(A)$.

    An \emph{instance}~$I$ of a relation~$R$ is a partial map
    \[ I \colon \N \rightharpoonup \bigtimes_{A \in \sort(R)} \Dom(A) \]
    with finite domain of definition $\Def(I) := \{ j \in \N \mid
    I(j) \text{ is defined} \}$.
\end{definition}

The above definition of a relation instance allows for duplicate tuples and preserves the order of the tuples.

For simplification, in this work we assume that $\Dom(A) = \N$ for all $A \in \sort(R)$, i.e., each tuple consists of positive integers.  Thus an instance~$I$ of a relation~$R$ of arity~$m$ can be seen as just a partial map $I \colon \N \rightharpoonup \N^m$.

\begin{definition}
    Let~$R$ be a relation and~$I$ an instance of~$R$.  For a subset $K := \{ A_1, \dots, A_s \} \subseteq \sort(R)$ of attributes, we have the projection $\pi_K(I)$ of the tuples in~$I$ to the attributes in~$K$.  Moreover, for a subset $J \subseteq \Def(I)$ we have the restricted map $I|_J \colon J \to \N^m$ defined for all indices $j \in J$.
\end{definition}

Given an instance~$I$ of a relation~$R$ we denote by $t_j[A_k]$ the value of the attribute~$A_k$ in the $j$-th tuple $t_j := I(j)$ of~$I$.

Similarly, for a subset $K = \{ A_1, \dots, A_s \} \subseteq \sort(R)$ of attributes we denote by $t_j[A_1 \ldots A_s]$ the tuple of values $t_j[A_1], \dots, t_j[A_s]$, that is, the $j$-th tuple $\pi_K(I(j))$ of the projection $\pi_K(I)$.

\begin{definition}
    A \emph{functional dependency} in a relation~$R$ is a pair $f := (\{ A_1, \dots, A_s \}, B)$, where $A_1, \dots, A_s, B \in \sort(R)$ are attributes.  We write such a pair as $A_1 \ldots A_s \to B$.
    
    An instance~$I$ of a relation~$R$ is said to \emph{fulfill} the functional dependency~$f$ (write $I \models f$) if for all $j_1, j_2 \in \Def(I)$ it holds 
    \[ t_{j_1}[A_1 \ldots A_s] = t_{j_2}[A_1 \ldots A_s] ~\Rightarrow~ t_{j_1}[B] = t_{j_2}[B] . \]
    
    Moreover, the instance~$I$ \emph{fulfills} a set~$F$ of functional dependencies (write $I \models F$) if $I \models f$ for all $f \in F$.
\end{definition}

We also consider relation instances having unspecified values at some positions.  Let $\var$ be a (countably infinite) set of variables.

\begin{definition}
    Let~$I$ be an instance of a relation~$R$ with attributes $\sort(R) = \{ A_1, \dots, A_m \}$.  A \emph{position} in~$I$ is a pair $(j, A_k)$ with $j \in \Def(I)$ and $k \in \{ 1, \dots, m \}$; it represents the cell of attribute~$A_k$ in the $j$-th tuple, with value $t_j[A_k]$.  The instance obtained by putting the value~$v$ at position $p = (j, A_k)$ is denoted by $I_{p \from v}$.
    
    Let $Q = \{ q_1, \dots, q_k \}$ be a set of positions, $X = (x_1, \dots, x_k)$ distinct variables, and $V = (v_1, \dots, v_k)$ values in $\N$.  The instance obtained by replacing each position~$q_i$ by the variable~$x_i$ resp.\ by value~$v_i$ for $1 \le i \le k$ is denoted by $I_{Q \from X}$ resp.\ $I_{Q \from V}$ (hence, the instance obtained by replacing the positions from~$Q$ by variables, and then position~$p$ by the value~$v$, is $(I_{Q \from X})_{p \from v}$).
    
    An instance~$I$ containing distinct variables at positions $Q = \{ q_1, \dots, q_k \}$ \emph{fulfills} a set~$F$ of functional dependencies (write $I \models F$) if there exists a set of values $V = (v_1, \dots, v_k)$ such that $I_{Q \from V} \models F$.
    For one functional dependency $f \in F$ we write $I \models f$ if $I \models \{ f \}$.
\end{definition}

For $f := A_1 \ldots A_s \to B$ (and $I$ possibly containing variables) we readily obtain that $I \models f$ holds if and only if for all $j_1, j_2 \in \Def(I)$ such that $t_{j_1}[B] \notin \var$ and $t_{j_2}[B] \notin \var$ there holds
\[ t_{j_1}[A_1 \ldots A_s] = t_{j_2}[A_1 \ldots A_s] ~\Rightarrow~ t_{j_1}[B] = t_{j_2}[B] . \]
So the above definition requires a single functional dependency to be fulfilled only for tuples without variables.  Indeed, since all variables are distinct, it is always possible to set values in their positions so that the functional dependency is fulfilled.

We note that if~$I$ is an instance with variables and~$F$ is a set of functional dependencies, then $I \models F$ is in general not equivalent to $I \models f$ for all $f \in F$.  However, it can be shown that $I \models F$ if and only if $I \models f$ for all $f \in F^*$, where~$F^*$ is the transitive closure of~$F$.  This equivalence ensures the same semantics as in the original work by Arenas and Libkin~\cite{DBLP:conf/pods/ArenasL03}, and we assume that the transitive closure of functional dependencies is provided.

\subsection{Entropy and Information Content}\label{ssec:entropy}

In order to define the information content we need to recall some basic notions from information theory, cf.~\cite{CovTho}.
We deal with discrete probability spaces $\mcX = (X, P_{\mcX})$ on finite sets~$X$, where $P_{\mcX}(x)$ for $x \in X$ denotes the probability of the event $\{ x \}$.  The information-theoretic \emph{entropy} of the probability space~$\mcX$ is given by \[ H(\mcX) \,:=\, - \sum_{x \in X} P_{\mcX}(x) \log P_{\mcX}(x) . \]

If $\mcA = (A, P_{\mcA})$ and $\mcB = (B, P_{\mcB})$ are probability spaces on finite sets~$A$ and~$B$ with joint distribution $P_{\mcA \times \mcB}$, then the \emph{conditional probability} of $a \in A$ given $b \in B$ is defined by
\[ P(a | b) \,:=\, \frac{P_{\mcA \times \mcB}(a, b)} {P_{\mcB}(b)} \]
provided that $P_{\mcB}(b)$ is non-zero. 
Conversely, the conditional probabilities $P(a | b)$ with the probabilities $P_{\mcB}(b)$ determine the joint distribution $P_{\mcA \times \mcB}$ by the above formula.

\begin{definition}
    Let $\mcA = (A, P_{\mcA})$, $\mcB = (B, P_{\mcB})$ be probability spaces.
    The \emph{conditional entropy} of~$\mcA$ given~$\mcB$ is
    \[ H(\mcA \mid \mcB) \,:=\, - \sum_{b \in B} P_{\mcB}(b)
    \sum_{a \in A} P(a | b) \log P(a | b) . \]

    This value describes the remaining uncertainty in probability space~$\mcA$ given the outcome in~$\mcB$.    
    If the probability distributions of~$\mcA$ and~$\mcB$ are independent, then $P(a | b) = P_{\mcA}(a)$ for all $b \in B$, and thus we obtain $H(\mcA \mid \mcB) = - \sum_{a \in A} P_{\mcA}(a) \log P_{\mcA}(a) = H(\mcA)$.
\end{definition}

Consider now a relation instance~$I$ of arity~$m$ and denote by $\Pos := \Def(I) \times \{ A_1, \dots, A_m \}$ its set of positions.  Let~$F$ be a set of functional dependencies fulfilled by~$I$ and $p \in \Pos$ a position.  We define two probability spaces as follows.

First, let $\mcB(I, p) := (\mcP(\Pos \setminus \{ p \}), P_{\mcB})$, where~$P_{\mcB}$ is the uniform distribution on the set of all subsets of $\Pos \setminus \{ p \}$.
This space models the possible cases when we lose a set of possible values from the instance~$I$ on positions other than the considered one~$p$.

Then, for $k \in \N$ we let $\mcA_F^k(I, p) := ( \{ 1, \dots, k \}, P_{\mcA})$, where the conditional probability of $v \in \{ 1, \dots, k \}$ given $Q \subseteq \Pos \setminus \{ p \}$ is
\[ P(v | Q) \,:=\, \begin{cases} 1 / \# V_Q & \text{if } v \in V_Q , \\ 0 & \text{otherwise} , \end{cases} \]
with $V_Q := \{ v \in \{ 1, \dots, k \} \mid (I_{Q \from X})_{p \from v} \models F \}$.
This probability space models the possible values in $\{ 1, \dots, k \}$ to be put in at position~$p$ for which we lost the value, when~$Q$ is the set of positions of the other lost values in the instance.

\begin{definition}
    Let~$F$ be a set of functional dependencies for a relation~$R$ and let~$I$ be an instance of~$R$ with $I \models F$.  The \emph{information content} of position~$p$ with respect to~$F$ in instance~$I$ is given as
    \[ \INF_I(p \mid F) \,:=\, \lim_{k \to \infty} \frac{\INF_I^k(p \mid F)} {\log k} , \] where $\INF_I^k(p \mid F) := H(\mcA_F^k(I, p) \mid \mcB(I, p))$ is the conditional entropy of the probability space modeling the possible values for the considered position~$p$ given the space modeling the possible sets of other lost values in the instance.
\end{definition}

Unfortunately, when using the above formula for the conditional entropy $\INF_I^k(p \mid F)$ directly, the computation grows exponentially with the number of cells in the given instance.  Indeed, for each cell except the one at position~$p$ the value can be deleted or not, so every subset of $\Pos \setminus \{ p \}$ is an elementary event in the probability space $\mcB(I, p)$.  Therefore, each additional cell in the instance~$I$ doubles the number of events to be taken into account for the computation of the information content.

\subsection{Simplifications}\label{ssec:simply}

We now provide more compact and simplified (but still exponentially complex) formulas for the information content.  The first result easily follows from the definition of conditional entropy.

\begin{proposition}
    Let~$F$ be a set of functional dependencies and~$I$ an instance with $I \models F$.  Then the information content of a position~$p$ in~$I$ with respect to~$F$ is given by 
    \[ \INF_I(p \mid F) = \frac 1 {2^{\#\Pos - 1}} \sum_{\! Q \subseteq \Pos \setminus \{ p \} \!} \lim_{k \to \infty} \frac{\log \# V_Q} {\log k} , \]
    where $V_Q := \{ v \in \{ 1, \dots, k \} \mid (I_{Q \from X})_{p \from v} \models F \}$.
\end{proposition}

\begin{proof}
    By the definition of conditional entropy we obtain
    \begin{align*}
        \INF_I^k (p \mid F) &= - \sum_{\! Q \subseteq \Pos\setminus\{p\} \!} P_{\mcB}(Q) \sum_{v=1}^k P(v|Q) \log P(v|Q) \\
        &= - \sum_{\! Q \subseteq \Pos\setminus\{p\} \!} \frac 1 {\#\mcP(\Pos\setminus\{p\})} \sum_{v \in V_Q} \frac 1 {\# V_Q} \log \frac 1 {\# V_Q} \\
        &= \frac 1 {2^{\#\Pos-1}} \sum_{\! Q \subseteq \Pos\setminus\{p\} \!} \log \# V_Q .
    \end{align*}
    Normalizing and computing the limit shows the assumption.
\end{proof}

The values $\INF_I(p \mid F, Q) := \lim_{k \to \infty} \frac{\log \# V_Q} {\log k}$ can be seen as the information content of~$p$ in~$I$ with respect to~$F$ given a fixed subset $Q \subseteq \Pos \setminus \{ p \}$ to be substituted by variables.
The next result shows that for these, there are only two possible outcomes.

\begin{lemma}
    Let~$F$, $I$ and~$p$ be as above.  For any $Q \subseteq \Pos \setminus \{ p \}$ we have $\INF_I(p \mid F, Q) \in \{ 0, 1 \}$.
\end{lemma}

\begin{proof}
    Consider any \enquote{fresh} value $a \in \N$, which does not appear in the column of position~$p$ in the relation instance~$I$.  It is easy to see that, whether the instance $(I_{Q \from X})_{p \from a}$ fulfills~$F$ or not, does not depend on the choice of those values~$a$.  Therefore,  as $k \to \infty$ we either have $\log \# V_Q / \log k \to 1$ or $\log \# V_Q / \log k \to 0$.
\end{proof}

From this lemma and its proof, we deduce the following simplification for computing information content.

\begin{proposition}\label{prop:simply}
    Let~$F$ be a set of functional dependencies and~$I$ an instance with $I \models F$.  Let~$p$ be a position in~$I$ with attribute~$A$, and let $a \in \N$ be any value that does not appear in the column of attribute~$A$.  This yields:
    \[ \INF_I(p \mid F) = \frac{ \# \{ Q \subseteq \Pos \setminus \{ p \} \mid (I_{Q \from X})_{p \from a} \models F \} } { 2^{\# \Pos - 1} } \]
\end{proposition}

\begin{proof}
    Define $\mcP_a := \{ Q \subseteq \Pos \setminus \{ p \} \mid (I_{Q \from X})_{p \from a} \models F\}$.
    From the proof of the previous lemma we have that $\INF_I(p \mid F,Q) = 1$ if and only if $Q \in \mcP_a$.
    Thus we obtain \[
        \INF_I(p \mid F) = \frac 1 {2^{\#\Pos-1}} \sum_{\! Q \subseteq \Pos \setminus \{ p \} \!} \INF_I(p \mid F,Q) = \frac {\# \mcP_a} {2^{\#\Pos-1}} . \qedhere \]
\end{proof}

Although we have somewhat simplified the computation of information content, one still needs to consider all subsets of $\Pos \setminus \{ p \}$, leading to exponential complexity.  In the following, we therefore present several optimizations for this computation.

\section{Optimizations}
\label{sec:opt}

In this section, we provide optimizations to speed up the computation of information content.  In Section~\ref{ssec:size} we deal with exact methods and prove their correctness, while in Section~\ref{ssec:mc} we present an approximation.

\subsection{Reducing the Problem Size}\label{ssec:size}

We give two shortcuts for computing the information content in a relation instance.  The first identifies those positions where there is not any redundancy, so that the information content equals~$1$.

\begin{definition}
    Let $p = (j, B) \in \Pos$ be a position with attribute~$B$ in an instance~$I$ and let~$f$ be a functional dependency $A_1 \dots A_s \to B$ with $I \models f$.  We say that the value at~$p$ is \emph{unique} with respect to~$f$ if for every $j' \in \Def(I)$ there holds
    \[ t_j[A_1 \ldots A_s] = t_{j'}[A_1 \ldots A_s] ~\Rightarrow~ j = j' \]
    (so if $j \ne j'$, then $t_j[A_1 \ldots A_s] \ne t_{j'}[A_1 \ldots A_s]$).
    
    For a set~$F$ of functional dependencies with $I \models F$, the value at~$p$ is \emph{unique} with respect to~$F$ if it is unique with respect to all $f \in F$ of the form $A_1 \ldots A_s \to B$.
\end{definition}

Note that in particular, a value at position~$p$ with attribute~$B$ is unique with respect to~$F$ in case the attribute~$B$ does not appear on the right-hand side of any functional dependency in~$F$.

\begin{proposition}\label{prop:ones}
    Let $p \in \Pos$ be a position in an instance~$I$, where $I \models F$ for a set of functional dependencies~$F$.  Then $\INF_I(p \mid F) = 1$ if and only if the value at~$p$ is unique with respect to~$F$.
\end{proposition}

\begin{proof}
    Let $p = (j, B) \in \Pos$ and let~$a$ be a \enquote{fresh} value not appearing in the column of attribute~$B$ in instance~$I$.

    \enquote{$\Leftarrow$}.
    By Prop.~\ref{prop:simply} it suffices to show that for every subset\break $Q \subseteq \Pos \setminus \{ p \}$ it holds that $(I_{Q \from X})_{p \from a} \models F$.  So let $f \in F$ be a functional dependency $A_1 \ldots A_s \to B'$.
    From $I \models f$ we know that $I_{Q \from X} \models f$.  Then we have to show that 
    \[ (I_{Q \from X})_{p \from a} \models f , \]
    i.e., if $t_{j_1}[A_1 \ldots A_s] = t_{j_2}[A_1 \ldots A_s]$ for some $j_1, j_2 \in \Def(I)$ it still holds that $t_{j_1}[B'] = t_{j_2}[B']$, even after inserting value~$a$ at position~$p$.
    
    Suppose first that $B \ne B'$.  If $B \notin \{ A_1, \dots, A_s \}$ (and $B \ne B'$) the assertion above is clear for any value~$a$, since the statement is not affected.
    On the other hand, if $B = A_i$ for some $1 \le i \le s$, then by inserting the fresh value~$a$ the hypothesis becomes false, so the statement remains valid.

    Now consider the case $B = B'$.  We may assume that one of $j_1, j_2$ equals the index~$j$ of position~$p$ and write~$j'$ for the other index.  Then if $t_j[A_1 \ldots A_s] = t_{j'}[A_1 \ldots A_s]$, it follows from the uniqueness property that $j = j'$.  So the above statement still holds after inserting the value~$a$, since the tuple indices coincide.

    \enquote{$\Rightarrow$}.
    Suppose the value at~$p$ is not unique with respect to~$F$.  Then there is a functional dependency $A_1 \dots A_s \to B$ in~$F$ such that $t_j[A_1 \ldots A_s] = t_{j'}[A_1 \ldots A_s]$ for an index~$j'$ distinct from the index~$j$ of position~$p$.
    However, considering $Q := \varnothing$ this functional dependency is not fulfilled anymore by the instance $(I_{Q \from X})_{p \from a} = I_{p \from a}$.  Indeed, as the value~$a$ does not appear in the column of attribute~$B$, we have $t_j[B] \ne t_{j'}[B]$ after the insertion.
    This shows that $\INF_I(p \mid F) < 1$ by Prop.~\ref{prop:simply}.    
\end{proof}

In the second optimization, we reduce the considered instance to the relevant tuples and attributes.  This may reduce the number of cells, and thus decrease the runtime exponentially.
After performing this step, we can apply the first shortcut in the smaller table and use the outcome for the original instance.

Let~$I$ be a relation instance of arity~$m$ and let $J \subseteq \Def(I)$,\break $K \subseteq \{ A_1, \dots, A_m \}$ where the $A_k$ are the attributes of the relation.  The subinstance $I(J, K)$ consists of all tuples of the projection $\pi_K(I)$ with index $j \in J$.
We also let $\Pos(J, K) := J \times K \subseteq \Pos$ be the corresponding set of positions.

\begin{proposition}\label{prop:subtable}
    Let~$F$ be a set of functional dependencies and~$I$ an instance with $I \models F$.  Denote by~$J_0$ the set of all indices $j_0 \in \Def(I)$ such that for some position $p = (j_0, A_k) \in \Pos$ the value at~$p$ is not unique with respect to~$F$.  Let~$K_0$ the union of all attributes $\{ A_1, \dots, A_s, B \}$ involved in a functional dependency $A_1 \dots A_s \to B$ in~$F$.  Then for every $J \supseteq J_0$ and $K \supseteq K_0$ there holds
    \[ \forall p \in \Pos(J, K) \colon \INF_I(p \mid F) = \INF_{I(J, K)}(p \mid F) . \]
\end{proposition}

\begin{proof}
    Write $I' := I(J, K)$ and $\Pos' := \Pos(J, K)$.  Let $p = (j, B) \in \Pos'$ be a position in the subtable and let~$a$ be a \enquote{fresh} value not appearing in the column of attribute~$B$ in the whole instance~$I$.
    Consider a functional dependency $f \colon A_1 \ldots A_s \to B$ in~$F$, as well as subsets $Q' \subseteq \Pos' \setminus \{ p \}$ and $Q \subseteq \Pos \setminus \{ p \}$ such that $Q \cap \Pos' = Q'$.  We claim that
    \[ (I'_{Q' \from X})_{p \from a} \models f \quad\Leftrightarrow\quad (I_{Q \from X})_{p \from a} \models f \,. \]
    Indeed, suppose that $t_j[A_1 \ldots A_s] = t_{j'}[A_1 \dots A_s]$ for some $j, j' \in \Def(I)$ in the instance $(I_{Q \from X})_{p \from a}$.  If both $j, j' \in J$ this implies $t_j[B] = t_{j'}[B]$ since $(I'_{Q' \from X})_{p \from a}$ fulfills~$f$.  In the other case, at least one index is not in~$J$, say~$j$.  By assumption, the value of the $j$-th tuple at attribute~$B$ is unique with respect to~$f$, hence it follows that $j = j'$ and so $t_j[B] = t_{j'}[B]$.  The converse direction is clear.

    Let $r := \# \Pos$ and $r' := \# \Pos'$.  Given $Q' \subseteq \Pos' \setminus \{ p \}$ one easily verifies the number of subsets $Q \subseteq \Pos \setminus \{ p \}$ with $Q \cap \Pos' = Q'$ to be $2^{r - r'}$ (one adjoins a subset in $\Pos \setminus \Pos'$).  Then letting
    \begin{gather*}
    \mcP_a := \{ Q \subseteq \Pos \setminus \{ p \} \mid (I_{Q \from X})_{p \from a} \models F \} \,, \\
    \mcP'_a := \{ Q' \subseteq \Pos' \setminus \{ p \} \mid (I'_{Q' \from X})_{p \from a} \models F \} \,,
    \end{gather*}
    we deduce from these observations that
    \[ 2^{r - r'} \# \mcP'_a = \# \mcP_a \,. \]
    Therefore, using Prop.~\ref{prop:simply} it follows that
    \[ \INF_I(p \mid F) = \frac {\# \mcP_a} {2^{r-1}} = \frac {2^{r-r'} \# \mcP'_a} {2^{r-1}} = \frac {\# \mcP'_a} {2^{r'-1}} = \INF_{I'}(p \mid F) . \qedhere \]
\end{proof}

\begin{example}
\label{example:fds}
    Consider the instance \smallskip
    \begin{center}  
      \begin{tabular}{cccc}\toprule
        A & B & C & D \\\midrule
        7 & 2 & 8 & 4 \\
        5 & 2 & 8 & 6 \\
        7 & 2 & 8 & 6 \\\bottomrule
      \end{tabular}
    \end{center} \smallskip
    and the set of functional dependencies $F := \{ \text A \to \text C \}$.

    Since the attributes~$\text A$,~$\text B$ or~$\text D$ do not appear on the right-hand side of the functional dependency, Prop.~\ref{prop:ones} implies that $\INF_I(p \mid F) = 1$ for all $p = (j, A_k)$ with $A_k \ne \text C$. 
    Additionally, by Prop.~\ref{prop:ones} we obtain that $\INF_I((2, \text C) \mid F) = 1$, since the value at position $p = (2, \text C)$ is unique with respect to~$F$.

    We can reduce the table using Prop.~\ref{prop:subtable} by removing the second tuple and the attributes~$\text B$ and~$\text D$.  The resulting subtable is \smallskip
    \begin{center}
      \begin{tabular}{cccc}\toprule
        A & C \\\midrule
        7 & 8 \\
        7 & 8 \\\bottomrule
      \end{tabular}
    \end{center} \smallskip
    for which the number of subsets in $\Pos(J, K) \setminus \{ p \}$ is reduced from $2^{15}$ to $2^3$, i.e., by a factor over 4,000.
    
    Consider position $p = (1, \text C)$ and the subsets $Q \subseteq \Pos(J, K) \setminus \{ p \}$.
    For $Q = \varnothing$, i.e., the values in all positions different from $(1, \text C)$ are present, the instance $(I_{Q \from X})_{p \from v} = I_{p \from v}$ fulfills the functional dependency $\text A \to \text C$ only if $v = 8$.  For all other subsets~$Q$, i.e., at least one more value is lost, we obtain $(I_{Q \from X})_{p \from v} \models F$ for all values~$v$.  Therefore, $\INF_I(p \mid F) = \frac 7 8 = 0.875$.  The computation of $\INF((3, \text C) \mid F)$ is similar, and we get 
    \[ \big( \INF_I((j, A_k) \mid F) \big) = \begin{pmatrix}
      1 & 1 & 0.875 & 1 \\
      1 & 1 & 1 & 1 \\
      1 & 1 & 0.875 & 1
    \end{pmatrix} . \]
\end{example}

\subsection{Monte Carlo Approximation}\label{ssec:mc}

Next, we present an approximation of the information content. 
Despite the aforementioned optimizations, the computation of information content may be practically out of reach, even in relation instances of moderate size.  This may be the case, in particular, if the prerequisites for applying the shortcuts are not given.  A workaround for these cases can probably be a good approximation.
This approach is complementary to the previous optimizations for computing exact values, e.g., by identifying cells with full information content first, then reducing the problem to a smaller subtable (if possible), and finally computing the values on this subtable with an approximation algorithm.

The approximation is computed with a randomized algorithm, the Monte Carlo method, as introduced in \cite{mitzenmacher2017}.  Instead of considering all subsets of positions (minus the considered one), we pick a sample of subsets uniformly at random and take the average of the results.  The following tool is deduced from Theorem~4.14 of \cite{mitzenmacher2017} and is useful for estimating the accuracy of this method.

\begin{proposition}[Hoeffding's inequality]
    Let $X_1, \dots, X_n$ be independent identically distributed real random variables such that $a_i \le X_i - \E [X_i] \le b_i$.  Then for any $c > 0$ it holds
    \[ \Pr \Big( \textstyle\sum\limits_{i=1}^n (X_i - \E [X_i]) \ge c \Big) \le \exp \Big( \dfrac {-2 c^2} {\sum_{i=1}^n (b_i - a_i)^2} \Big) . \]
\end{proposition}

To approximate the information content, we consider the probability space $\mcB(I, p) := (\mcP(\Pos \setminus \{ p \}), P_{\mcB})$ from Section~\ref{sec:prelim} with uniform distribution~$P_{\mcB}$.  Define random variables
\[ X_i \colon \mcP(\Pos \setminus \{ p \}) \to [0, 1 ] , \quad Q \mapsto \lim_{k \to \infty} \frac {\log \# V_Q} {\log k} , \]
with $V_Q := \{ v \in \{ 1, \dots, k \} \mid (I_{Q \from X})_{p \from v} \models F \}$.  Then by Prop.~\ref{prop:simply} there holds $\E [X_i] = \INF_I(p \mid F)$.  
This motivates the next theorem on the randomized approach to compute the information content with accuracy~$\eps$ and confidence $1 \!-\! \delta$.

\begin{theorem}
    Let $p \in \Pos$ be a position in a relation instance~$I$, where $I \models F$ for a set~$F$ of functional dependencies.  Let $X_1, \dots, X_n$ be independent identically distributed random variables as above and $X := \frac 1 n \sum_{i=1}^n X_i$ their average.  Then for all $\eps, \delta > 0$ it holds
    \[ \Pr \big( \vert X - \INF_I(p \mid F) \vert \ge \eps \big) \le \delta \]
    provided that $n \ge 2 \ln(2 / \delta) / \eps^2$.
\end{theorem}

\begin{proof}
    Since $0 \le X_i \le 1$ we have $-1 \le X_i - \E [X_i] \le 1$.  Therefore, by applying Hoeffding's inequality with $\sum_{i=1}^n (b_i - a_i)^2 = 4 n$ we find for $\Pr(\sum_{i=1}^n (X_i - \E [X_i]) \ge n \eps)$ the upper bound $\exp(-n \eps^2 / 2)$.  Consequently, we have
    \[ \Pr(\vert X - \E [X_i] \vert \ge \eps) \le 2 \exp(-n \eps^2 / 2) \le \delta . \qedhere \]
\end{proof}

\begin{figure}[t]
  \includegraphics[width=\linewidth]{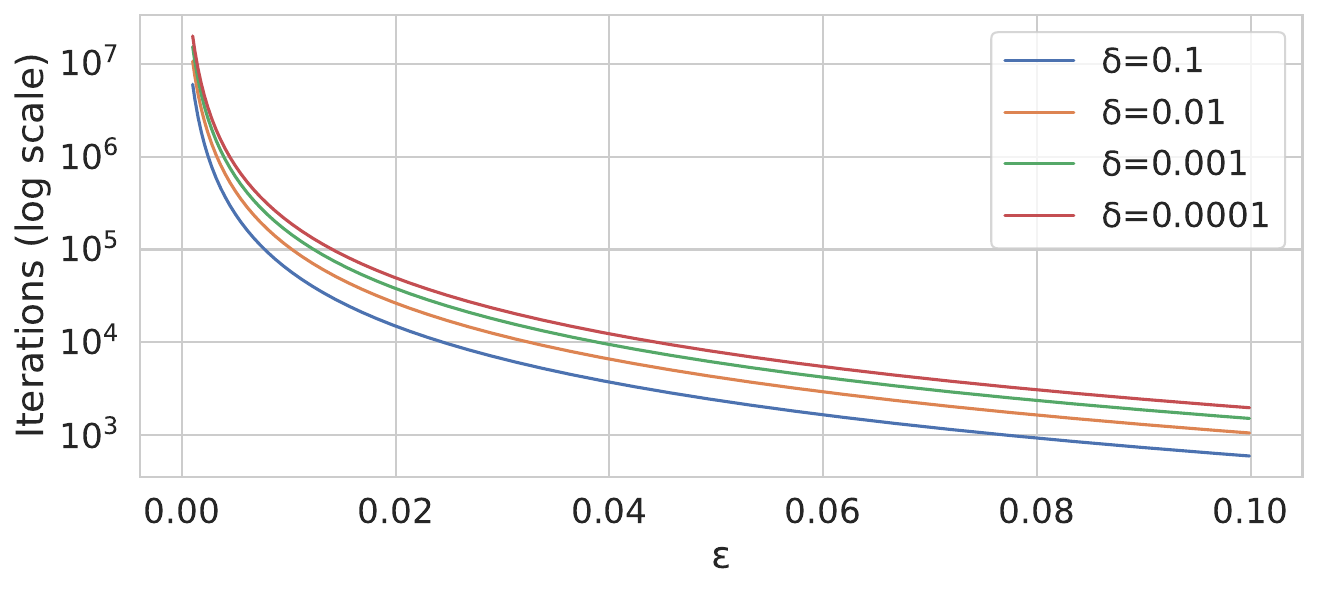}
  \caption{Required iterations to achieve an accuracy ($\eps$) with a certain confidence ($1 \!-\! \delta$) in Monte Carlo approximation.}
  \label{fig:accuracy}
\end{figure}

\begin{example}
    Assume that for the approximation of information content for an instance~$I$ at position~$p$ with respect to~$F$, we allow an error of at most~$0.001$ with probability at least $99.9 \%$.  This can be achieved by sampling $Q \subseteq \Pos \setminus \{ p \}$ and computing $X_i(Q)$ at least $2 \ln(2 / 10^{-3}) / (10^{-3})^2 \ge 1.52 \!\cdot\! 10^7$ times.
    
    If a less exact approximation is sufficient, say with an error of at most~$0.01$ with the same probability as before, then the required number of runs is lowered by a factor~$100$, hence only $1.52 \!\cdot\! 10^5$ samples are necessary.
\end{example}

Figure~\ref{fig:accuracy} shows a plot of the iterations required for reaching a certain accuracy~($\eps$) and a certain confidence~($1 \!-\! \delta$). 
For instance, to achieve an accuracy of 0.04 with a confidence of 99.9 \% ($1 \!-\! \delta$) we need around 10,000 iterations.

\input{experiments}

\section{Related Work}
\label{sec:related}

This article is an extended version of~\cite{shortversion} and contains proofs of all theorems mentioned in the shorter version. Here, we also provide experiments on more datasets and more detailed statistics and considerations on the entropies and their computation runtimes.

\paragraph{Functional Dependencies}

Dependency discovery is an established and active field, and we refer to~\cite{DBLP:series/synthesis/2018Abedjan} for an overview.
Visualizing dependencies is not as well studied, and existing approaches, such as sunburst diagrams or graph-based visualizations~\cite{10.1145/3132847.3133180} do not take the data instance into account.
In contrast, our plaque test does not visualize the dependencies per se, but the redundancies captured by them in the data.

In our visualization, we impose heat maps over relational data. Heat maps have also been explored elsewhere, to visualize the frequency of updates~\cite{DBLP:conf/cidr/BleifussBKNS19}.

\paragraph{Information Theory and Databases}

We build on an information-theoretic framework developed by Arenas and Libkin~\cite{DBLP:conf/pods/ArenasL03}, which aligns with classic normalization theory. Notably, we are not aware of any earlier implementations of this framework.

In an orthogonal effort, Lee~\cite{1702145} proposed entropies at the instance level which are not tied to a set of functional dependencies. Therefore, this approach does not lend itself to the plaque test proposed here.

More recently, information theory has also been investigated within database theory research in the context of deciding query containment~\cite{10.1145/3375395.3387645}, which is an orthogonal use case.

\section{Outlook}
\label{sec:outlook}

We propose a plaque test to visualize redundancies in relational data based on functional dependencies. 
As our discussion of real-world examples shows, the presence of plaque reveals interesting redundancies in the data. As we have seen, the plaque may single out prominent dependencies, reveal the need for schema normalization, or expose data as meaningless.

In our experiments, we work with data sets of only a modest size.
Scaling to larger datasets is a challenge for future work, and we see optimization opportunities via parallelization or dynamic programming.
Further, as we point out, the visual impact of inaccuracies introduced by Monte Carlo approximation depends on the difference between minimum and maximum entropy values. In future work, this could serve as a heuristic to automatically select a suitable number of iterations.

Once we are able to scale to larger datasets, we will need to address the consumability of our visualization. For example, we may allow users to interactively cluster cells with plaque for easy browsing and inspection. 

We may allow users to examine the effect of individual dependencies by excluding them from the visualization in an exploratory manner.
 This raises the question whether the entropies can be computed incrementally.

We also plan to investigate the visualization of other kinds of dependencies, such as join dependencies, since they are also covered by the underlying information-theoretic framework.

\balance

\medskip\noindent
{\sffamily\textbf{Acknowledgments.}}
This work was partially funded by \emph{Deut\-sche For\-schungsgemein\-schaft} (DFG, German Research Foundation) grant \#385808805.
We thank Meike Klettke for her feedback on an earlier version of this article.

\bibliographystyle{ACM-Reference-Format}
\bibliography{references}

\end{document}

%% file: example_figure.tex
\begin{figure}[b]
\centering

\definecolor{babyblue}{rgb}{0.54, 0.81, 0.94} 
\definecolor{ballblue}{rgb}{0.13, 0.67, 0.8} 
\definecolor{blue(munsell)}{rgb}{0.0, 0.5, 0.69} 
\definecolor{blue(ncs)}{rgb}{0.0, 0.4, 0.74} 

\begin{subfigure}[b]{0.49\textwidth}
\centering
\footnotesize 
\begin{tabular}%
{@~c@~c@~c@~c@~c@~c@~c@~}
\toprule
\underline{ID} & Album & Band & BYear & RYear & \underline{Track} & Title
\\
\midrule
1	
&Not That Kind
&Anastacia	
&1999 
&2000	
&1	
&Not That Kind
\\
1	
&Not That Kind
&Anastacia	
&1999 
&2000	
&2
&I’m Outta Love
\\
1	
&Not That Kind	
&Anastacia	
&1999	
&2000	
&3
&Cowboys \dots
\\
2	
&Wish You Were Here
&Pink Floyd	
&1965	
&1975	
&1	
&Shine On You\dots
\\
3	
&Freak of Nature	
&Anastacia	
&1999	
&2001	
&1	
&Paid my Dues
\\
\bottomrule
\end{tabular}
\caption{The relational input data.}
\label{fig:relation}
\end{subfigure}\medskip

\begin{subfigure}[b]{0.49\textwidth}
\centering
\footnotesize 
\setlength\tabcolsep{4.5pt}
\begin{tabular}%
{ccccccc}
\toprule
~\underline{ID}~ & \!Album\! & \,Band\, & BYear & RYear & \underline{Track} & Title
\\
\midrule
1
&\cellcolor{babyblue} 0.8
&\cellcolor{babyblue} 0.8
&\cellcolor{blue(munsell)} 0.6	
&\cellcolor{babyblue} 0.8	
&1	
&1
\\
1
&\cellcolor{babyblue} 0.8
&\cellcolor{babyblue} 0.8
&\cellcolor{blue(munsell)} 0.6	
&\cellcolor{babyblue} 0.8	
&1	
&1
\\
1	
&\cellcolor{babyblue} 0.8	
&\cellcolor{babyblue} 0.8
&\cellcolor{blue(munsell)} 0.6
&\cellcolor{babyblue} 0.8
&1	
& 1
\\
1
&1
&1
&1	
&1	
&1
& 1
\\
1 
&1
&1
& \cellcolor{ballblue} 0.7
&1	
&1	
&1
\\
\bottomrule
\end{tabular}
\caption{Entropies for 6 unary FDs.}
\label{fig:relation2}
\end{subfigure}\medskip

\begin{subfigure}[b]{0.49\textwidth}
\centering
\footnotesize 
\setlength\tabcolsep{4.5pt}
\begin{tabular}%
{ccccccc}
\toprule
~\underline{ID}~~ & \!Album\! & \,Band\, & BYear & RYear & \underline{Track} & Title
\\
\midrule
\cellcolor{blue(munsell)} 0.6
&\cellcolor{blue(munsell)} 0.6
&\cellcolor{blue(ncs)} 0.4	
&\cellcolor{blue(ncs)} 0.4	
&\cellcolor{blue(munsell)} 0.6	
&1	
&1
\\
\cellcolor{blue(munsell)} 0.6	
&\cellcolor{blue(munsell)} 0.6
&\cellcolor{blue(ncs)} 0.4	
&\cellcolor{blue(ncs)} 0.4	
&\cellcolor{blue(munsell)} 0.6
& 1
&1
\\
\cellcolor{blue(munsell)} 0.6
&\cellcolor{blue(munsell)} 0.6
&\cellcolor{blue(ncs)} 0.4	
&\cellcolor{blue(ncs)} 0.4	
&\cellcolor{blue(munsell)} 0.6	
&1	
&1
\\
1	
&1	
&1	
&1	
&1	
&1	
&1
\\
1	
&1	
&\cellcolor{ballblue} 0.7
&\cellcolor{ballblue} 0.7	
&1	
&1	
&1
\\
\bottomrule
\end{tabular}
\caption{Entropies for 23 unary FDs.}
\label{fig:relation3}
\end{subfigure}

\caption{Plaque tests for the original relation (top) with genuine functional dependencies (middle) and automatically discovered functional dependencies (bottom). Cell color/hue corresponds to entropy values.}
\label{fig:cds}
\end{figure}

%% file: experiments.tex
\section{Experiments}
\label{sec:experiments}

Our experiments target the following research questions:
\begin{description}

\item[RQ1] Is the plaque test useful for real-world datasets?

\item[RQ2] Can we afford to compute \emph{exact} entropy values, or do we need the Monte Carlo approximation?

\item[RQ3] How does the runtime of the Monte Carlo approximation scale with the number of iterations on real-world datasets?

\end{description}

In the following, we report on our insights with a first prototype and five real-world datasets. 

\subsubsection*{Implementation}
Our prototype implements our algorithms as a single-threaded Java implementation. As a dispatcher, we use a Python script that also measures end-to-end runtimes.

\subsubsection*{Datasets}
We explore five real-world datasets.
All functional dependencies are left-reduced with a single attribute on the right and were discovered by Metanome~\cite{DBLP:journals/pvldb/PapenbrockBFZN15}. 

Our first dataset describes natural satellites and originates from the WDC Web Table Corpus\footnote{\url{http://webdatacommons.org/webtables/index.html\#results-2015}}.
Metanome finds 35 functional dependencies, we analyze the first 150 rows.
We study four additional datasets\footnote{Retrieved from \url{https://hpi.de/naumann/projects/repeatability/data-profiling/fds.html}, on June 02, 2023}: The adult dataset (where we analyze the first 150 rows)  with 78 functional dependencies captures census data. The echocardiogram dataset (all 132 rows) with 538 functional dependencies describes heart attack patients.
The dataset NCVoter (first 150 rows) with 758 functional dependencies contains data about citizens and their voting behavior.
The iris dataset (first 150 rows) with 4 functional dependencies is about classification methods.

\subsubsection*{Environment}
Our server has an Intel Xeon Gold 6242R (3.1 GHz) CPU and 192GB of RAM, and runs Ubuntu~22.04 with Java~18. 

\newsavebox{\measurebox}
\begin{figure}[t]
\centering
\begin{tabular}[c]{cc}
\multicolumn{2}{c}{
\begin{subfigure}[t]{0.45\textwidth}
\includegraphics[width=\textwidth]{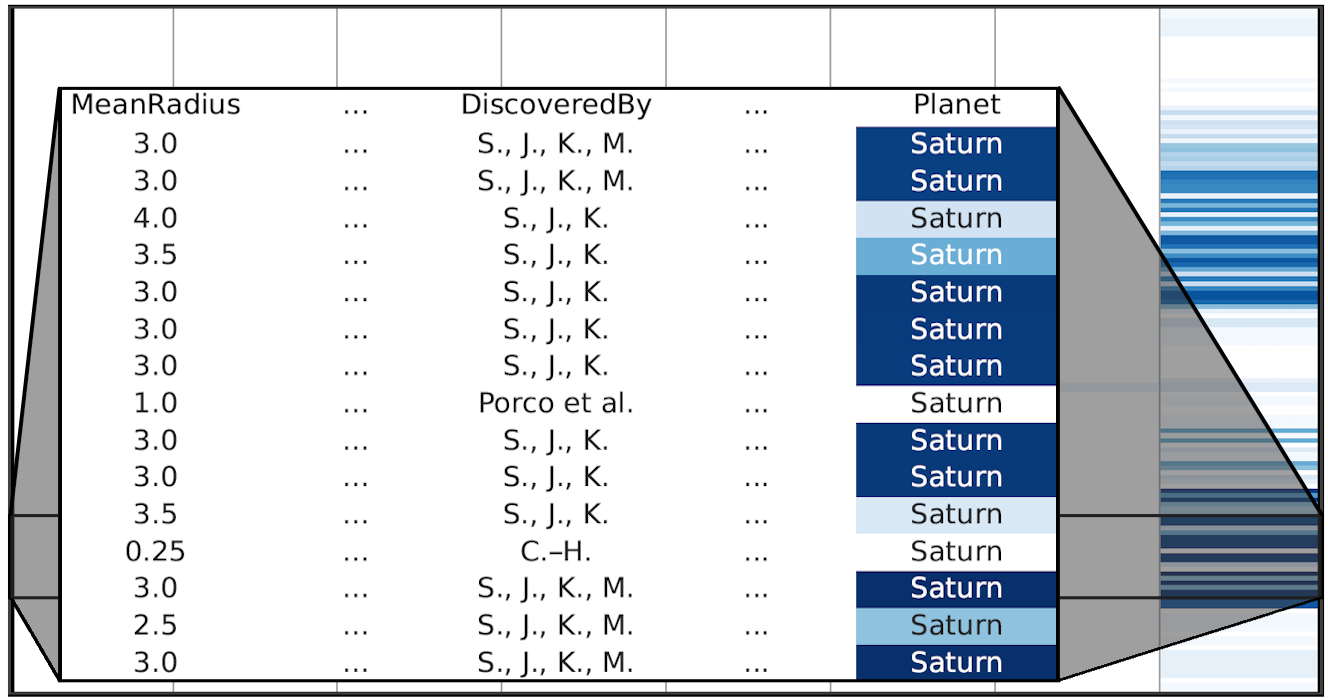}    
\caption{Satellites (150, Min: 0.61)}
\label{fig:plaque_satellite}
\end{subfigure}} \\
\begin{subfigure}[]{0.21\textwidth}
\includegraphics[width=\textwidth]{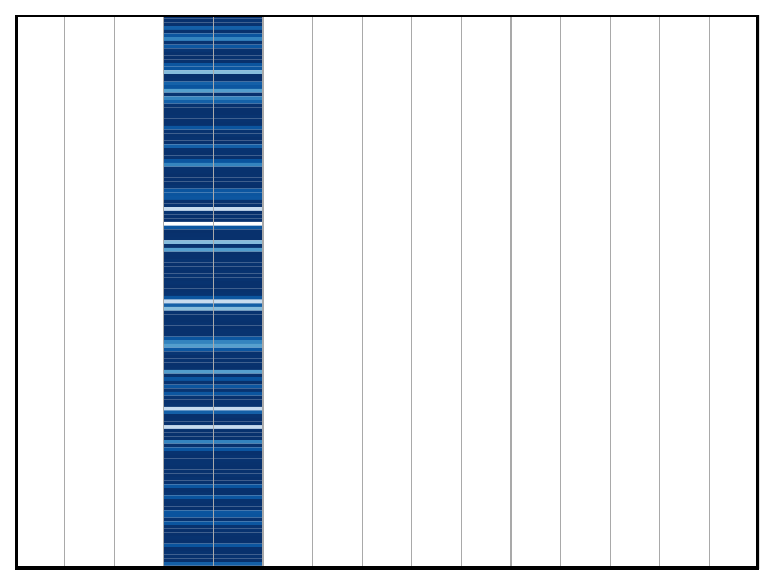}
\caption{Adult (150, Min: 0.50)} 
\label{fig:plaque_adult}
\end{subfigure}  
&
\begin{subfigure}[]{0.21\textwidth}
\includegraphics[width=\textwidth]{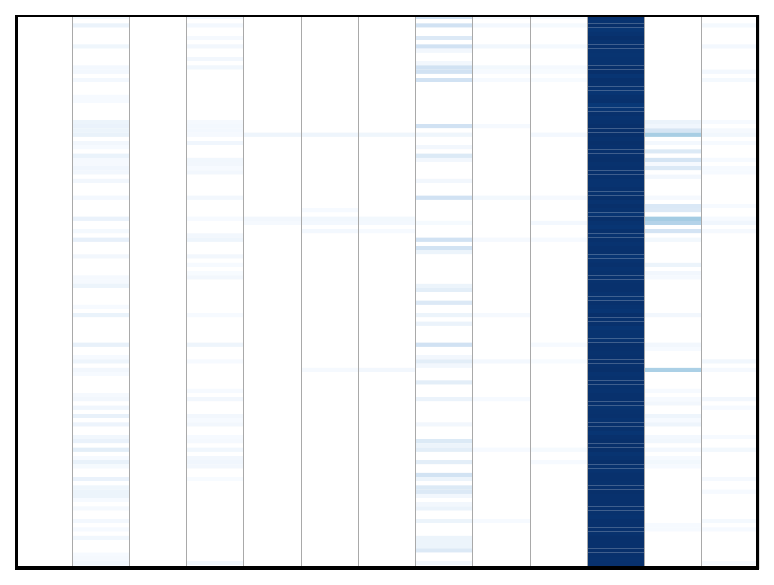}
\caption{Echoc. (132, Min: 0.00)} 
\label{fig:plaque_echo}
\end{subfigure} \\
\begin{subfigure}[]{0.21\textwidth}
\includegraphics[width=\textwidth]{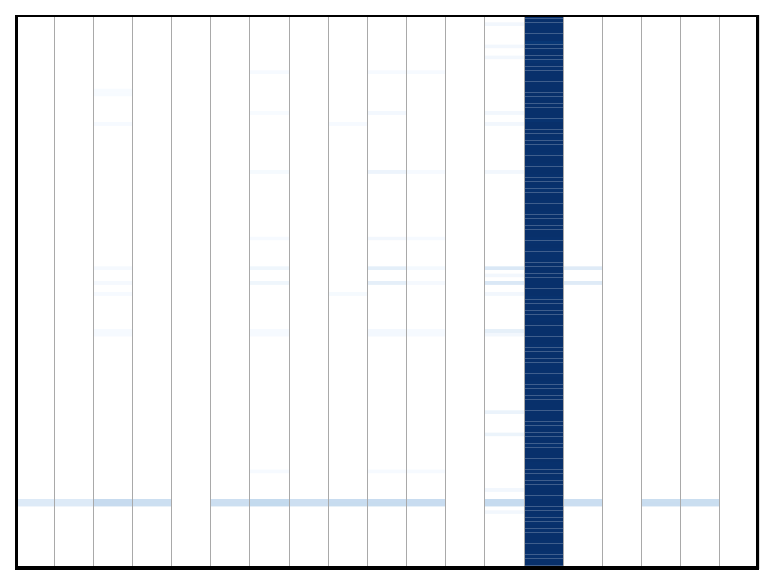}
\caption{NCVot. (150, Min: 0.00)}
\label{fig:plaque_voter}
\end{subfigure}
&
\begin{subfigure}[]{0.21\textwidth}
\includegraphics[width=\textwidth]{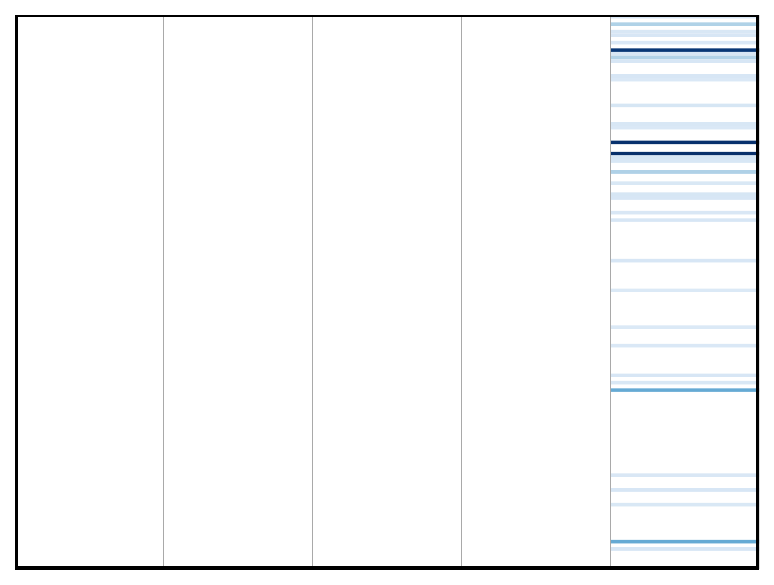}
\caption{Iris (150, Min: 0.95)}
\label{fig:plaque_iris}
\end{subfigure}
\end{tabular}
\caption{\enquote{Plaque tests} applied to real-world data. The sub-captions state the numbers of rows analyzed and minimum entropy values computed (rounded). The color scale is normalized individually with respect to the minimum entropy. The zoom-in in Subfigure~(a) highlights  a subset of rows.}
\label{fig:plaque-test}
\end{figure}

\medskip\noindent
We now address the research questions in turn.

\medskip\noindent{\sffamily\textbf{RQ1.}}
We computed the plaque tests for the five datasets and show the results in Figure~\ref{fig:plaque-test}. Specifically, 
the entropy values are computed with Monte Carlo simulation set to 100,000 iterations, thus accuracy of approx.\ 0.01 with a confidence of 99\%. 

Cells with an entropy value of~1 are shown in white, and cells with values below~1 are colored in blue, where darker shades indicate lower entropy values. The color scales are calibrated individually, so colors cannot be compared between datasets. 
We discuss each dataset in turn.

\paragraph{Satellites}

Figure~\ref{fig:plaque_satellite} shows the plaque test applied to the satellite dataset. Plaque is locally concentrated, as only the column~\enquote{Planet} and very few cells in column~\enquote{Notes} contain entropy values below~1.
We zoom in on a subset of the rows, omitting cells with an entropy value of~1 and showing the values for columns \enquote{MeanRadius}, \enquote{DiscoveredBy}, and \enquote{Planet}. 

First, we present a possible interpretation of the visual effect.
For tuples with a mean radius of 3.0, the entropy of the attribute planet is the lowest.
A mean radius of 3.0 occurs only for Saturn satellites. 
Closer inspection reveals that \enquote{MeanRadius} is on the left-hand side of several functional dependencies (as discovered by Metanome), including \enquote{$\text{MeanRadius, DiscoverdBy} \to \text{Planet}$}. 
By inspecting the causes for low entropy values in this fashion, data analysts can gain insight into why automated tools discover certain dependencies. They can then decide whether the dependency is genuine or merely an artifact of the data.

Next, we discuss the dataset w.r.t.\ our proposed optimizations.
In Figure~\ref{fig:plaque_satellite}, the majority of cells has entropy~1. Figure~\ref{fig:histogram} shows a histogram of the entropy values, to provide more insight into their distribution.
Among a total of 1,200 cells, approx.\ 90\% have an entropy value of~1. Values below~1 are scarce, with the lowest value being close to~0.6. Overall, the information content of the cells is rather high and only around 5\% of the cells have an entropy value below~0.9.

The high number of cells with entropy~1 shows that the optimizations introduced in Section~\ref{sec:opt} effectively reduce the problem size. In this exemplary dataset, 1,083 out of 1,200 cells have an information content of~1, thus the computation can be skipped for 90\% of the dataset by applying the first optimization introduced in Section~\ref{ssec:size}. This is a reduction by a factor~10. Since there are 35 rows with full information content, by applying the second optimization the number of rows in the dataset can be reduced by this number, thus by 280 cells. This reduces the number of computations by a factor of $2^{280}$, which is larger than $10^{84}$. Overall, the introduced optimizations reduce the computation effort for the satellites dataset by a factor over $10^{85}$.  \smallskip

\paragraph{Adult} 

Figure~\ref{fig:plaque_adult} shows the plaque test applied to the census data. Only two columns, \enquote{education} and \enquote{education-num} have entropy values below 1. Moreover, in each row, both columns have the \emph{same} entropy value.
Closer inspection reveals that there are functional dependencies \enquote{$\text{education-num} \to \text{education}$} as well as  \enquote{$\text{education} \to \text{education-num}$}. This causes the respective entropy values to agree. 

In consequence, a data analyst might decide to decompose this relation into the second normal form, 
by storing the mapping between
\enquote{education} and \enquote{education-num}
in a separate relation.

In fact, normalization theory was a main motivator behind the original work by Arenas and Libkin.

\paragraph{Echocardiogram and NCVoter}

Figure~\ref{fig:plaque_echo} shows the plaque test applied to the patient data,
Figure~\ref{fig:plaque_voter} the plaque test for voter data.
Among the five datasets, these datasets
have the highest number of columns with entropy values below~1: This affects 11 out of 13 columns in Echocardiogram and 15 out of 19 columns in NCVoter, but mostly only sparsely. 

One column in the echocardiogram dataset stands out, where all entropy values are zero (rounded to one digit after the decimal point).
Inspection reveals that this is the column that originally contained the patient's name, which was changed to a single global string constant as a means of anonymization. Consequently, for every attribute, there is an obvious functional dependency, with this attribute on the right-hand side.
Our plaque test correctly reveals that this column literally has almost no informational value.

In the dataset NCVoter, there is also a column with entropy values of zero.
The corresponding attribute is the state of the voting person. But since the data captures voters in only one state, namely North Carolina (NC), the column has a single-valued domain.
Therefore, as in the echocardiogram dataset, this column is functionally dependent on every other column, and the plaque test confirms that it does not have any information content.

\paragraph{Iris}

Figure~\ref{fig:plaque_iris} shows the plaque test applied to a classification data set.
In this dataset, there are only FDs with the last column, \enquote{class}, on the right side. This means that a class cannot uniquely define the value of any other attribute. Consequently, only this column has entropy values below 1 and is therefore a candidate for a redundant attribute.

\subsubsection*{Results}

We applied our visual plaque tests to standard datasets used in dependency discovery research. The plaque tests appear to be helpful in data exploration: When \enquote{plaque} is detected, we can always find an intuitive explanation for its causes. 

In the examples discussed, it highlighted a prominent functional dependency, revealed a good opportunity for schema normalization, and exposed data with nearly no informational value.

Moreover, the plaque test is very selective:
The test is strongly positive for only a few attributes.
This sparsity of cells that test strongly positive for plaque is particularly notable in the case of the echocardiogram dataset, despite the over 500 automatically discovered dependencies. Compared to this high number of dependencies, the result of the plaque test is easily consumable, and data analysts are visually directed towards the most pertinent redundancies.

\begin{figure}[t]
    \includegraphics[width=\linewidth]{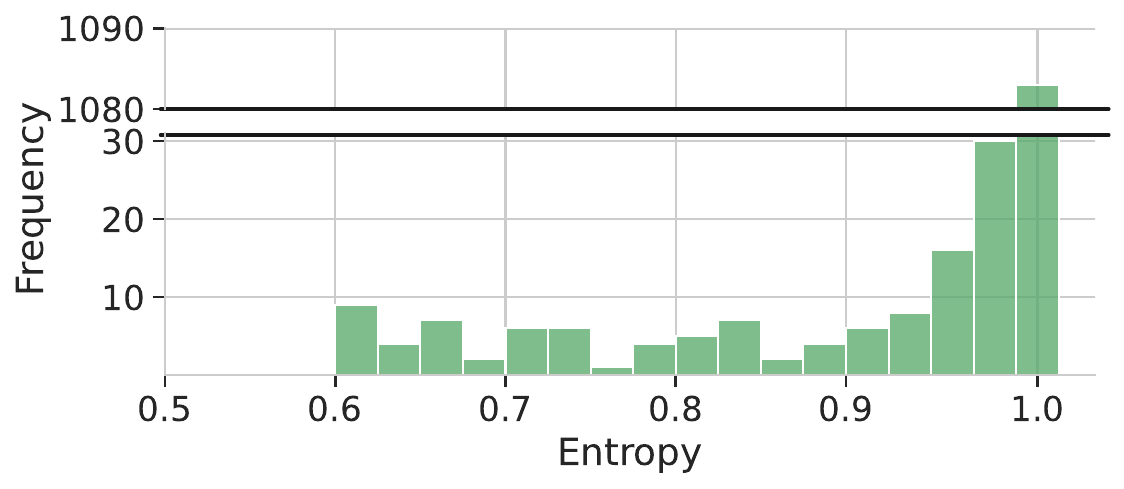}
    \caption{Histogram over entropy values in the first 150 rows of the satellites dataset (accuracy: 0.01, 99.9\% confidence).}
    \label{fig:histogram}
\end{figure}

\begin{table}
\setlength\extrarowheight{-2pt}
\caption{Runtimes in seconds for computing exact entropy values w/ and w/o the optimizations from Section~\ref{sec:opt} for the first~$i$ rows of satellite data. Runs marked \enquote{-} were aborted after 24 hours.}
\small
\centering
\begin{tabular}{rrr}
    \toprule
    \textbf{\#Rows} & \textbf{Unoptimized} & \textbf{Optimized} \\
    \midrule
    1 & 0.128 & 0.097 \\
    2 & 1.318 & 0.099 \\
    3 & 461.059 & 0.320 \\
    4 & - & 0.355 \\
    5 & - & 25,221.186 \\
    6 & - & - \\
    \bottomrule
\end{tabular}
\label{tab:runtime}
\end{table}

\medskip\noindent{\sffamily\textbf{RQ2.}}
Table~\ref{tab:runtime} shows the runtime in seconds for computing the entropy values on subsets of the satellite data. 
We compute exact entropies and do not yet apply the Monte Carlo approximation.
We compare the algorithm with the optimizations from Section~\ref{ssec:size} disabled/enabled.

The unoptimized algorithm can process only three rows in 24 hours. Using the optimizations, up to five rows could be computed in 24 hours. 

\subsubsection*{Results}

Although the optimizations are effective compared to the unoptimized implementation, computing the \emph{exact} entropies remains prohibitively expensive.

\begin{figure}[t]
    \includegraphics[width=\linewidth]{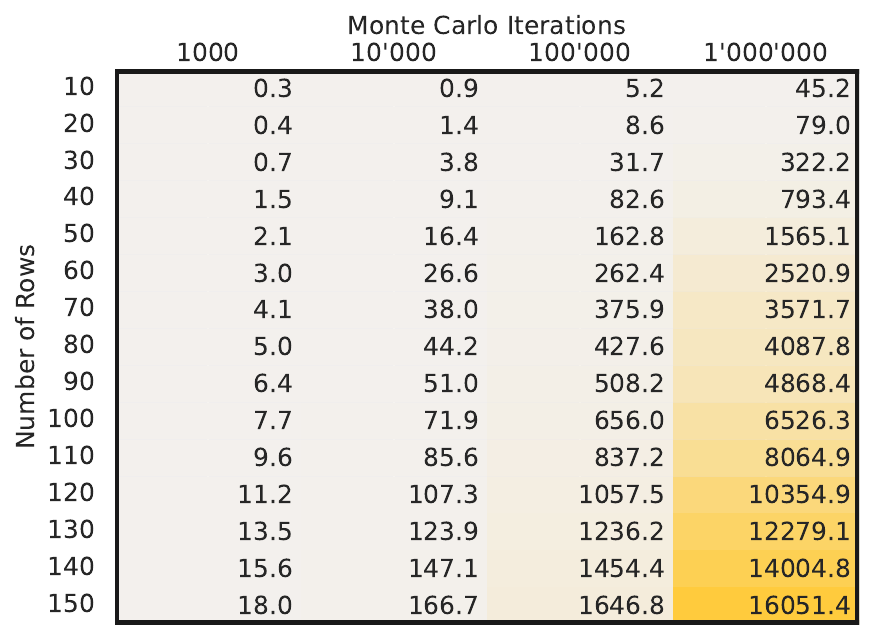}
    \caption{Runtime in seconds on the satellite dataset, for different numbers of Monte Carlo iterations and subset sizes of satellite data. Higher saturation indicates longer runtimes.}
    \label{fig:runtime}
\end{figure}

\begin{figure}[t]
\centering
\begin{tabular}[c]{cc}
\begin{subfigure}[]{0.21\textwidth}
\includegraphics[width=\textwidth]{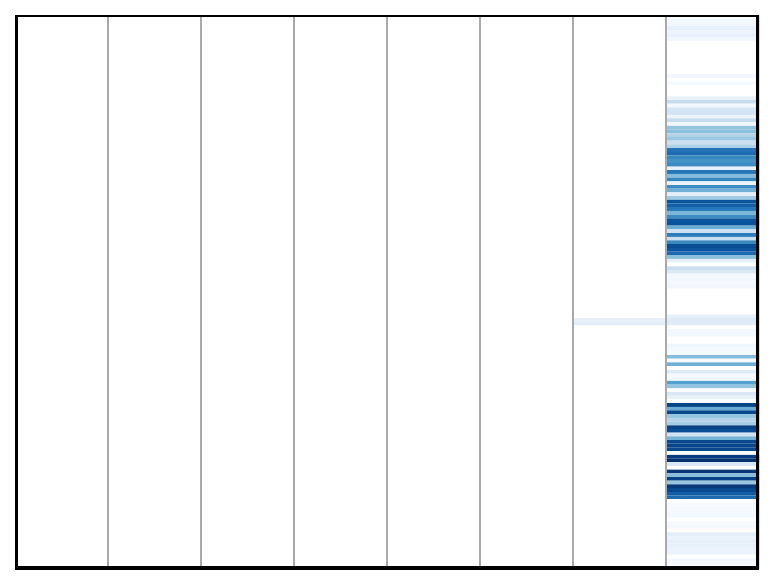}
\caption{1,000 iterations} 
\label{fig:satellites-mc-1000}
\end{subfigure}  
&
\begin{subfigure}[]{0.21\textwidth}
\includegraphics[width=\textwidth]{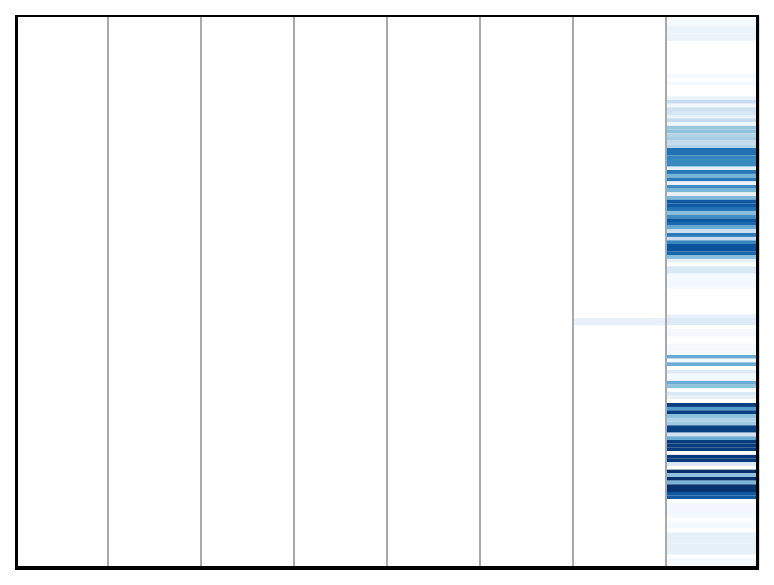}
\caption{1,000,000 iterations} 
\label{fig:satellites-mc-1000000}
\end{subfigure}
\end{tabular}
\caption{\enquote{Plaque tests} applied to the satellites dataset using the Monte Carlo approximation for different numbers of iterations.}
\label{fig:satellites-different-mcs}
\end{figure}

\medskip\noindent{\sffamily\textbf{RQ3.}}
We next explore the Monte Carlo approximation, in combination with our optimizations from Section~\ref{ssec:size}.

Figure~\ref{fig:runtime} shows the runtimes in seconds for different subsets of the satellite data and iterations of the Monte Carlo method. 

For reasonably large subsets, runtime scales linearly with the number of iterations, while input size influences runtime more heavily. Calculating the entropies for 150 rows takes around 4.5 hours at 1,000,000 iterations. We reach an accuracy of approx.\ 0.01 with 99\% confidence at 100,000 iterations (see Figure~\ref{fig:accuracy}), which takes around 30 minutes to calculate.

Figure~\ref{fig:satellites-different-mcs} shows the visualization of the entropies, computed with 1,000 and 1,000,000 iterations, respectively. Despite this disparity in the number of iterations and the expected effect on accuracy, there are only minuscule differences in the generated images. This suggests that we may use a smaller number of Monte Carlo iterations to reduce runtime considerably (in this case by three orders of magnitude) with negligible impact on the visualization. 
In fact, using the results of the entropy computation for 1,000 and 1,000,000 iterations, the maximum entropy difference between the corresponding cells from these two computations is about 0.048. Most differences are distinctly smaller. Of 117 cells with an entropy below 1, for only 9 cells, the difference exceeded a value of 0.02.
However, as we calibrate the colors based on the minimum and maximum values of a table, the visual impact of deviations in entropy values depends on the distance between these values and may be more noticeable if this distance is small.

\subsubsection*{Results}

The approximation greatly improves the runtime behavior.
Since we use the entropy values as a basis for visualization, small deviations in accuracy lead to equally small deviations in the color scale. 
The differences will most likely not be discernible to the human eye.

However, scalability to larger inputs remains a challenge and will require further improvements. 
We point out ideas for future optimizations as part of our outlook in Section~\ref{sec:outlook}.